\documentclass{revtex4}
\usepackage{amsmath,amsthm}

\usepackage{graphicx}
\newcommand{\bra}[1]{\langle#1|}
\newcommand{\ket}[1]{|#1\rangle}
\usepackage{amsfonts}
\newtheorem{theorem}{Theorem}[section]

\newtheorem{proposition}[theorem]{Proposition}

\theoremstyle{remark}
\newtheorem{remark}[theorem]{Remark}

\theoremstyle{definition}
\newtheorem{definition}[theorem]{Definition}

\theoremstyle{example}

\theoremstyle{notation}

\begin{document}
\title{Generalized Bargmann functions, their growth and von Neumann lattices}
\author{A. Vourdas $^{1,3}$, K. A. Penson $^2$, G. H. E. Duchamp $^3$, and A. I. Solomon $^{2,4}$}
\address{$^1$ Department of Computing, University of Bradford, Bradford BD7 1DP, UK\\
$^2$ Laboratoire de Physique Th\'{e}orique de la Mati\`{e}re Condens\'{e}e,
Universit\'{e} Pierre and Marie Curie, CNRS UMR 7600, Tour 13 5-i\`{e}me et., B.C. 121, 4 place Jussieu,
F 75252 Paris Cedex 05, France\\
$^3$ Universit\'{e} Paris 13, LIPN, Institut Galil\'{e}e, CNRS UMR 7030, 99 Av.J.B. Clement,
F 93430 Villetaneuse, France\\
$^4$ The Open University, Physics and Astronomy Department, Milton Keynes MK7 6AA, UK}

\begin{abstract}
Generalized Bargmann representations which are based on generalized coherent states are considered.
The growth of the corresponding analytic functions in the complex plane is studied.
Results about the overcompleteness or undercompleteness of discrete sets of these
generalized coherent states are given.
Several examples are discussed in detail.
\end{abstract}
\maketitle

\section{Introduction}

The Bargmann representation \cite{A1} represents  quantum states by
analytic functions in the complex plane.
This allows the powerful theory of analytic functions to be used in a quantum mechanical context.
An example of a result which can be proved only by the use of the theory of analytic functions
is related to the overcompleteness or undercompleteness of discrete sets of coherent states,
e.g. the von Neumann lattice of coherent states \cite{B1,B2,B3,B4,B5,B6,B7,B8,B9,B10,B11,B12,B13,B14}.
This is based on deep theorems relating the growth of analytic functions and the
density of their zeros\cite{C1,C2,C3}.
The study of the Bargmann representation determines the nature of admissible functions
belonging to the Hilbert space.
Furthermore the study of the paths of the zeros of the Bargmann functions under time evolution
has led to important physical insight for many systems \cite{D1,D2,D3,D4,D5,D6,D7,D8,D9}.

Many generalizations of coherent states \cite{E1,E2,E3} have been considered in the literature.
In the present paper we consider the generalizations studied in \cite{F1,F2,F3,F4,F5,F6,F7}
which have led to many interesting sets of generalized coherent states.
We use them to define generalized Bargmann representations, which represent the various quantum states by
analytic functions in the complex plane.
The requirement of convergence for the scalar product in these representations determines
the maximum growth of the generalized Bargmann functions and defines the corresponding Bargmann spaces.
Theorems that relate the growth of analytic functions in the complex plane to the density
of their zeros lead to results about the overcompleteness or undercompleteness of discrete sets of these
generalized coherent states.
The general theory is applied to four examples: the standard coherent states and three
examples of generalized coherent states.
Therefore we study explicitly three novel types of generalized Bargmann representations.

In section II we briefly review  known results on the growth of analytic functions
in the complex plane, and the relation to the density of their zeros.
In section III we define   generalized coherent states.
In section IV we introduce  generalized Bargmann functions and  study their growth.
We also show that a discrete set of these generalized coherent states
is overcomplete [{\it resp.} undercomplete] if its density is greater [{\it resp.} smaller] than a critical density.
In section V we discuss explicitly four examples.
We conclude in section VI with a discussion of our results.

\section{Growth of analytic functions in the complex plane and the density of their zeros}

Analytic functions are characterized by their growth and the density of their zeros and we
briefly summarize these concepts \cite{C1,C2,C3}.
Let $M(R)$ be the maximum modulus of an analytic function $f(z)$ for $|z|=R$.
Its growth is described by the order $r$ and the type $s$, which are defined as follows:
\begin{eqnarray}
r=\lim _{R\to \infty }\sup \frac {\ln \ln M(R)}{\ln R};\;\;\;\;\;\;
s=\lim _{R\to \infty }\sup \frac {\ln M(R)}{R^r}.
\end{eqnarray}
These definitions imply that  $M(R)\sim \exp(sR^r)$ as $R$ goes to infinity
(here the $\sim$ indicates that $M(R)$ is log-asymptotic to $\exp(sR^r)$).
\begin{definition}
\mbox{}
\begin{itemize}
\item[(1)]
${\mathfrak B}(r,s)$ is the set of analytic functions
in the complex plane with order smaller than $r$, and also functions with
order $r$ and type smaller or equal to $s$.
\item[(2)]
${\mathfrak B}_1(r,s)$ is the set of analytic functions
in the complex plane with order smaller than $r$, and also functions with
order $r$ and type smaller than $s$.
\end{itemize}
\end{definition}
${\mathfrak B}(r,s)-{\mathfrak B}_1(r,s)$ is the set of analytic functions
with order $r$ and type $s$.

We next consider the sequence $\zeta _1,...,\zeta _N,...$ where
$0<|\zeta _1|\le |\zeta _2|\le ....$ and the limit is infinite (as $N$ goes to infinity).
We denote by $n(R)$ the number of terms of this sequence within the circle $|z|<R$.
The density of this sequence is described by the numbers
\begin{eqnarray}
t=\lim _{R\to \infty }\sup \frac {\ln n(R)}{\ln R};\;\;\;\;\;
\overline \delta =\lim _{R\to \infty }\sup \frac {n(R)}{R^t};\;\;\;\;
\underline \delta =\lim _{R\to \infty }\inf \frac {n(R)}{R^t}.
\end{eqnarray}
In the cases considered in this paper the $\lim _{R\to \infty }\frac {n(R)}{R^t}$ exists
and therefore $\overline \delta =\underline \delta $. 
Below we will use the simpler notation $\delta =\overline \delta =\underline \delta $.  
These definitions imply that asymptotically $n(R)\sim \delta R^t$ as $R$ goes to infinity.
\begin{definition}
We say that the density $(t, \delta)$ of a sequence is smaller than $(t_1, \delta_1)$
and we denote this by  $(t, \delta)\prec (t_1, \delta_1)$ if
$t< t_1$ and also if $t=t_1$ and $\delta<\delta_1$ (lexicographic order).
\end{definition}
\begin{remark}\label{rem}
The density depends only on the absolute values of
$\zeta _n$, i.e. the sequences $\zeta _1,...,\zeta _N,...$
and $\zeta _1\exp (i\theta _1),...,\zeta _N \exp(i\theta _N),...$,
where $\theta _N$ are arbitrary real numbers,
have the same density.
Also if  we add or subtract a finite number of complex numbers in a sequence, its
density remains the same.
\end{remark}
An example of a sequence which has density $(t, \delta)$ is
\begin{eqnarray}\label{20}
\zeta _N=\exp \left [\frac{1}{t}\ln \left(\frac{N}{\delta}\right )+i\theta _N\right]
\end{eqnarray}
where $\theta _N$ are arbitrary real numbers.

The relationship between the growth of an analytic function in the complex plane and
the density of its zeros is described through the following inequalities \cite{C1,C2,C3}:
\begin{eqnarray}\label{100}
t\le r;\;\;\;\;\;\;\;sr\ge \delta.
\end{eqnarray}

\section{Generalized coherent states}

Let ${\cal H}$ be the Hilbert space
corresponding to a Hamiltonian operator $h$ and $\ket{n}$ its number eigenstates.
Following \cite{F1,F2,F3,F4,F5,F6,F7}, we consider the generalized coherent states
\begin{eqnarray}
\ket {z;\rho}=[{\cal N}_\rho (|z|^2)]^{-1/2}\sum _{n=0}^\infty \frac {z^n}{\rho (n)^{1/2}}\ket{n};\;\;\;\;\;
{\cal N}_\rho (|z|^2)=\sum _{n=0}^\infty \frac {|z|^{2n}}{\rho (n)}.
\end{eqnarray}
Here $\rho (n)$ is a positive function of $n$ with $\rho(0)=1$.
The series in the normalization constant ${\cal N}_\rho (|z|^2)$ converges within some disc $|z|<R\le \infty$.
{\bf In this paper we only consider cases where $\rho (n)$ is increasing fast enough as a function of $n$,
so that this series converges in the whole complex plane.}

The overlap of two generalized coherent states is
\begin{eqnarray}
\langle z';\rho\ket{z;\rho}=[{\cal N}_\rho (|z'|^2)]^{-1/2}[{\cal N}_\rho (|z|^2)]^{-1/2}{\cal K}_\rho (z'^*,z);\;\;\;\;\;
{\cal K}_\rho (\zeta,z)=\sum _{n=0}^\infty \frac{(\zeta z)^n}{\rho(n)}.
\end{eqnarray}
Here ${\cal K}_\rho (\zeta ,z)$ is the reproducing kernel. Clearly ${\cal K}_\rho (z^*,z)={\cal N}_\rho (|z|^2)$.

The choice $\rho (n)=n!$ is an example, and leads to the standard coherent states.
For $\rho (n)=n!$ the overlap $\langle z';\rho\ket{z;\rho}$ is everywhere 
different from zero.
For other $\rho (n)$ the ovelap might have zeros.

The resolution of the identity in terms of the generalized coherent states is
a weak operator equality given by
\begin{eqnarray}\label{res}
\int _{\mathbb C}d^2z{\widetilde W}_\rho(|z|^2)\ket {z;\rho}\bra {z;\rho}={\bf 1},
\end{eqnarray}
where ${\widetilde W}_\rho (x)>0$ is a function such that
\begin{eqnarray}
\int _0^\infty dx x^n W_\rho (x)=\rho (n);\;\;\;\;\;\;W_\rho (x)=\frac{\pi{\widetilde W}_\rho (x)}
{{\cal N}_\rho (x)}.
\end{eqnarray}
We stress that for arbitrary $\rho(n)$ the ${\widetilde W}_\rho (x)$ does not always exist.

The property of temporal stability states that
if we act with the evolution operator corresponding to a certain Hamiltonian $h$
on the generalized coherent states, we get other coherent states.
Our generalized coherent states have this property:
\begin{eqnarray}
\exp (i\tau h)\ket {z;\rho}=\ket {z\exp (i\tau \omega);\rho}.
\end{eqnarray}
For a study of states having the above property, see\cite{1000}.

\section{Generalized Bargmann functions}

Let $\ket{f}$ be an arbitrary state in the Hilbert space $H$:
\begin{eqnarray}
\ket{f}=\sum_{n=0}^\infty f_n\ket{n};\;\;\;\;\;\sum_{n=0}^\infty |f_n|^2=1.
\end{eqnarray}
We represent this state by the following analytic function in the complex plane:
\begin{eqnarray}\label{50}
F(z;\rho)=[{\cal N}_\rho (|z|^2)]^{1/2}\langle z^*;\rho\ket{f}=\sum_{n=0}^\infty \frac{f_n z^n}{[\rho(n)]^{1/2}}.
\end{eqnarray}
For $\rho (n)=n!$ this is the standard Bargmann function.
But other choices of $\rho(n)$ lead to generalizations of the Bargmann function.

If the resolution of the identity in Eq.(\ref{res}) exists, 
then it leads to the following expression for the scalar product of two
states $\ket{f}$ and $\ket {g}$ represented by the functions $F(z;\rho)$ and $G(z;\rho)$, correspondingly:
\begin{eqnarray}\label{12}
\langle g\ket{f}=(G,F)=\int _{\mathbb C}\frac{d^2z}{\pi}W_\rho(|z|^2)[G(z;\rho)]^*F(z;\rho)
=\sum_{n=0}^\infty g_n^*f_n.
\end{eqnarray}
It is known that a pointwise bound for $F(z;\rho)$ is $|F(z;\rho)|^2\le {\cal K}_\rho (z^*,z)(F,F)$ \cite{Hall}.

As an example, we consider the generalized coherent state $\ket{\zeta ;\rho}$ which is represented by the generalized Bargmann function
\begin{eqnarray}\label{88}
F_{\rm coh}(z;\rho)=[{\cal N}_\rho (|\zeta|^2)]^{-1/2}{\cal K}_\rho (\zeta,z).
\end{eqnarray}

\begin{definition}
The Bargmann space ${\cal B} (W_\rho)$ consists of analytic functions in the complex plane 
$F(z;\rho)$, such that $(F,F)$ is finite, with the scalar product given by Eq.(\ref{12}).
\end{definition}
Let ${\mathfrak S}$ be a set of quantum states in the space $H$.
${\mathfrak S}$ is called  a {\em total} set in
$H$, if there exists no state $\ket{s}$ in $H$
which is orthogonal to all states in ${\mathfrak S}$.

${\mathfrak S}$ is called  {\em undercomplete} in $H$, if there exists a state $\ket{s}$ in $H$
which is orthogonal to all states in ${\mathfrak S}$.

A total set
${\mathfrak S}$ in $H$ is called  {\em overcomplete} in $H$, if
there exist at least one state $\ket{u}$ in ${\mathfrak S}$
such that the ${\mathfrak S}-\{\ket{u}\}$ is also a total set.

A total set ${\mathfrak S}$ in H is called {\em complete} if every proper subset is not a total set. 
Below we will mainly use the terms overcomplete, complete and undercomplete.
\begin{proposition}
\mbox{}
\begin{itemize}
\item[(1)]
If
\begin{eqnarray}
W_\rho (|z|^2)\sim \exp \left [-2{\mathfrak b}(\rho)|z|^{{\mathfrak a}(\rho)}\right ]
\end{eqnarray}
as $|z|\rightarrow \infty$, then the set of functions in the generalized Bargmann space ${\cal B} (W_\rho)$
satisfies the
\begin{eqnarray}
{\mathfrak B}_1[{\mathfrak a}(\rho),{\mathfrak b}(\rho)] \subset {\cal B} (W_\rho) 
\subset {\mathfrak B}[{\mathfrak a}(\rho),{\mathfrak b}(\rho)].
\end{eqnarray}
Functions with order of growth ${\mathfrak a}(\rho)$ and type ${\mathfrak b}(\rho)$, 
might or might not belong to ${\cal B} (W_\rho)$.
\item[(2)]
Let $\{z_N\}$ be a sequence of complex numbers with density $(t, \delta)$.
The set of generalized coherent states $\ket {z_N;\rho}$ is overcomplete [{\it resp.} undercomplete] when
$(t, \delta)\succ [{\mathfrak a}(\rho),{\mathfrak b}(\rho){\mathfrak a}(\rho)]$ 
[{\it resp.} $t<{\mathfrak a}(\rho)$].
\end{itemize}
\end{proposition}
\begin{proof}
\mbox{}
\begin{itemize}
\item[(1)]
The integral in Eq.(\ref{12}) diverges for functions with order of growth greater than ${\mathfrak a}(\rho)$,
or for functions with order equal to ${\mathfrak a}(\rho)$ and type greater than ${\mathfrak b}(\rho)$.
Therefore the functions in the generalized Bargmann space 
${\cal B} (W_\rho)$, belong to the set ${\mathfrak B}[{\mathfrak a}(\rho),{\mathfrak b}(\rho)]$.

Also if the functions have order less than ${\mathfrak a}(\rho)$
or order equal to ${\mathfrak a}(\rho)$ and type less than ${\mathfrak b}(\rho)$, 
the integral in Eq.(\ref{12}) converges.
Therefore functions in the set ${\mathfrak B}_1[{\mathfrak a}(\rho),{\mathfrak b}(\rho)]$ 
belong to the generalized Bargmann space 
${\cal B} (W_\rho)$.

We next show with examples, that functions in 
${\mathfrak B}[{\mathfrak a}(\rho),{\mathfrak b}(\rho)]-{\mathfrak B}_1[{\mathfrak a}(\rho),{\mathfrak b}(\rho)]$
might or might not belong to ${\cal B} (W_\rho)$.
We consider the special case $W_\rho (|z|^2)=\exp(-2\lambda |z|^2)$, i.e., ${\mathfrak a}(\rho)=2$ and
${\mathfrak b}(\rho)=\lambda >0$.
We also consider the functions
\begin{eqnarray}
F_1(z;\rho)=\frac{\exp (\lambda z^2)-1}{z};\;\;\;\;\;\;F_2(z;\rho)=z^n\exp(\lambda z^2).
\end{eqnarray}
They both have growth with order $2$ and type $\lambda$ and belong to the space
${\mathfrak B}(2,\lambda )-{\mathfrak B}_1(2,\lambda)$.
The integral of Eq.(\ref{12}) converges with the first function and diverges with the second function.
Therefore the first function belongs to ${\cal B} (W_\rho)$, and the second function does not belong to it.

\item[(2)]
It follows from Eq.(\ref{50}) that if $F(\zeta;\rho)=0$ then the corresponding state ${\ket f}$
is orthogonal to the generalized coherent state $\ket{\zeta^*; \rho}$.

We consider a set of generalized coherent states $\ket {z_N;\rho}$
with density $(t, \delta)\succ [{\mathfrak a}(\rho),{\mathfrak b}(\rho){\mathfrak a}(\rho)]$.
If this is not a total set then there exists a function $F(z;\rho)$
in the set ${\mathfrak B}[{\mathfrak a}(\rho),{\mathfrak b}(\rho)]$,
which is equal to zero for all $z_N$.
But this is not possible because it violates the relations in Eq.(\ref{100}).
Therefore the set of states $\ket{z_N; \rho}$ is a total set in the space $H$.
The same result is also true if we subtract a finite number of terms from the sequence $\{z_N\}$.
Therefore the set of generalized coherent states $\ket{z_N; \rho}$ is overcomplete.

We next consider  a set of generalized coherent states $\ket {z_N;\rho}$
with density of the corresponding sequence $t<{\mathfrak a}(\rho)$.
In this case we can construct a state which is orthogonal to all $\ket {z_N;\rho}$ \cite{B12,B14}.
We use the Hadamard theorem\cite{C1,C2,C3} and consider the analytic function 
\begin{equation}\label{200}
F(z,\rho)=P(z) \exp [Q_q(z)],
\end{equation}
where
\begin{equation}
P(z)=z^m \prod _{N=1}^{\infty} E(z_N,p),
\end{equation}
\begin{equation}
E(A_N,0)=1-\frac {z}{z_N},
\end{equation}
\begin{equation}
E(A_N,p)=(1-\frac {z}{z_N})\exp \left
[\frac {z}{z_N}+\frac {z^2}{2z_N^2}+...+\frac {z^p}{pz_N^p}\right ].
\end{equation}
Here $E(A_N,p)$ are the Weierstrass factors,
$Q_q(z)$ is a polynomial of degree $q$, and $p$ is an integer. 
The maximum of $(p,q)$ is called the {\em genus} of
$F(z,\rho)$ and does not exceed its order.
$m$ is a non-negative integer and it is the multiplicity of the zero at the origin. 
The $z_N$ are clearly zeros of the $F(z,\rho)$ in
Eq.(\ref{200}). 

It remains to show that for some $Q_q(z)$ this function belongs to the generalized 
Bargmann space ${\cal B}(W_\rho)$.
We take $Q_q(z)=0$ and then the order of the growth of $F(z,\rho)=P(z)$ is equal to $t$ (theorem 7 in p.16 in \cite{C2}).
Therefore if $t<{\mathfrak a}(\rho)$ this function is indeed in the generalized Bargmann space ${\cal B}(W_\rho)$,
and the corresponding state is orthogonal to all
generalized coherent states $\ket{z_N; \rho}$. 
Consequently the set of these coherent states is undercomplete.
\end{itemize}
\end{proof}

\begin{remark}
In the `boundary case'  
where the density of the sequence $\{z_n\}$ has $t={\mathfrak a}(\rho)$
and $\delta\le {\mathfrak b}(\rho){\mathfrak a}(\rho)$, we can not state general results.
We mention some known results for special cases.

For example, when $t={\mathfrak a}(\rho)$ is non-integral, 
we consider the function $F(z; \rho)=P(z)$, as we did above.
This function has growth with order $t$ and its type $s$ satisfies
$s\le C\delta$ where $C$ is a constant that depends on the order $t={\mathfrak a}(\rho)$
(p.32 in \cite{C3}).
In this case for sequences $\{z_n\}$ with $\delta < {\mathfrak b}C^{-1}$ the 
function $F(z; \rho)=P(z)$ has growth less than $({\mathfrak a}(\rho),{\mathfrak b}(\rho))$ 
and it belongs to the Bargmann space.
Therefore  a set of generalized coherent states $\ket {z_N;\rho}$
with density of the corresponding sequence $(t,\delta)$ where $t$ is non-integral,
is undercomplete if $\delta$ is smaller than a critical value.
This result is not valid for integral $t$(p.32 in \cite{C3}).

Another example, is the set of standard coherent states which on a von Neumann lattice with $A=\pi$
is known to be overcomplete by one state \cite{E2}.
If we subtract a finite number of coherent states we get
an undercomplete set of standard coherent states
which is described by the same density.
This example shows two sequences with the same density (in the `boundary case'),
with corresponding coherent states forming an overcomplete and an undercomplete set.
\end{remark}

\section{Examples}

\subsection{$\rho_0(n)=n!$ : standard coherent states}
For the standard coherent states $\rho_0(n)=n!$ we have
\begin{eqnarray}
{\cal N}_{\rho _0}(|z|^2)=\exp \left (|z|^2\right );\;\;\;\;\;\;
W_{\rho _0}(|z|^2)=\exp (-|z|^2).
\end{eqnarray}
Therefore ${\mathfrak a}(\rho _0)=2$ and ${\mathfrak b}(\rho _0)=1/2$.
In this case 
${\mathfrak B}_1(2,1/2) \subset {\cal B} (W_\rho) \subset {\mathfrak B}(2,1/2)$.
The set of coherent states $\ket {z_n;\rho_0}$ is overcomplete or undercomplete in the cases that
the density $(t, \delta)$ of the sequence $\{z_n\}$ is
$(t, \delta)\succ (2,1)$ or
$t<2$, correspondingly.

An example is the rectangular von Neumann lattice $z_{NM}=A(N+iM)$ where $N,M$ are integers and $A$
is the area of each cell.
In this case $n(R)=\pi R^2/A$ and the density is described by $t=2$ and $\delta=\pi/A$.
Our results show that the set of coherent states $\{\ket {z_{NM};\rho _0}\}$ is overcomplete for
$A<\pi$. For this particular example, it is also known\cite{E2} that it is undercomplete for $A>\pi$.

We have explained in remark \ref{rem} that instead of the lattice $z_{NM}=A(N+iM)$
we can also use the complex numbers
$z_{NM}=A(N+iM)\exp (i\theta _{NM})$ with arbitrary phases $\theta _{NM}$.
The angular distribution of the zeros is totally irrelevant.

Another example is the sequence
\begin{eqnarray}
\zeta _N=\exp \left [\frac{1}{t}\ln ({2N})+i\theta _N\right]
\end{eqnarray}
where $\theta _N$ are arbitrary real numbers.
We have seen in Eq.(\ref{20}) that its density is $(t,1/2)$.
Therefore the set of the corresponding coherent states is overcomplete for
$t>2$ and undercomplete for $t<2$.
Also the sequence
\begin{eqnarray}
\zeta _N=\exp \left [\frac{1}{2}\ln \left(\frac{N}{\delta}\right )+i\theta _N\right]
\end{eqnarray}
where $\theta _N$ are arbitrary real numbers, has density $(2, \delta)$.
Therefore the set of the corresponding coherent states is overcomplete for $\delta >1$.

\subsection{$\rho_1(n)=(n!)^2$}\label{AA}
We consider the case that $\rho_1(n)=(n!)^2$. It was proved in \cite{F7} that
\begin{eqnarray}
{\cal N}_{\rho _1}(|z|^2)=I_0(2|z|);\;\;\;\;\;\;W_{\rho _1}(|z|^2)=2K_0(2|z|)
\end{eqnarray}
where $I_0$ and $K_0$ are modified Bessel functions of first and second kind, respectively.
But as $|z|\rightarrow \infty$
\begin{eqnarray}
2K_0(2|z|)\sim \left (\frac{\pi}{|z|}\right )^{1/2}\exp (-2|z|)
\end{eqnarray}
and therefore ${\mathfrak a}(\rho _1)=1$ and ${\mathfrak b}(\rho _1)=1$.
In this case 
${\mathfrak B}_1(1,1) \subset {\cal B} (W_\rho) \subset {\mathfrak B}(1,1)$.

The set of coherent states $\ket {z_N;\rho_1}$ is overcomplete [{\it resp,} undercomplete] when
the density $(t, \delta)$ of the sequence $\{z_N\}$ is
$(t, \delta)\succ (1,1)$
[{\it resp.} $t<1$].

An example is the {\bf one-dimensional} lattice
$z_{N}=\ell N\exp (i\theta _N)$, where $N$ is an integer and $\theta _N$ are arbitrary phases.
In this case $n(R)=2R/\ell$ and the density is described by $t=1$ and $\delta=2/\ell$.
Therefore the set of coherent states $\{\ket {z_{N};\rho _1}\}$ is overcomplete for
$\ell<2$.

\subsection{$\rho_2(n)=\frac{(n!)^3\pi^{1/2}}{2\Gamma(n+3/2)}$}\label{BB}
We consider the case  $\rho_2(n)=\frac{(n!)^3\pi^{1/2}}{2\Gamma(n+3/2)}$, where $\Gamma$ denotes the Gamma function, and for which \cite{F7}
\begin{eqnarray}
{\cal N}_{\rho_2}(|z|^2)=[I_0(|z|)]^2+2|z|I_0(|z|)I_1(|z|);\;\;\;\;\;\;W_{\rho _2}(|z|^2)=[K_0(|z|)]^2
\end{eqnarray}

But as $|z|\rightarrow \infty$
\begin{eqnarray}
[K_0(|z|)]^2\sim \left (\frac{\pi}{2|z|}\right )\exp (-2|z|)
\end{eqnarray}
Therefore ${\mathfrak a}(\rho _2)=1$ and ${\mathfrak b}(\rho _2)=1$.
In this case
${\mathfrak B}_1(1,1) \subset {\cal B} (W_\rho) \subset {\mathfrak B}(1,1)$.

Therefore our conclusions about the overcompleteness or undercompleteness of
the coherent states $\ket {z_N;\rho _2}$ are also valid for the coherent states here.
In connection with this it is interesting to compare the growth of $\rho(n)$ in these two cases,
as $n\;\rightarrow\; \infty$. We use the formula \cite{G1}
\begin{eqnarray}
\lim _{|z|\;\rightarrow\; \infty}\frac {\Gamma(z+a)}{\Gamma (z)}\exp (-a\ln z)=1
\end{eqnarray}
and from this we conclude that as $n\;\rightarrow\; \infty$
\begin{eqnarray}
\frac{\rho _2(n)}{\rho _1(n)}=\frac{\pi^{1/2}}{2}\frac{n!}{\Gamma(n+3/2)}\sim
\frac{\pi^{1/2}}{2}\exp \left (-\frac{1}{2}\ln n\right).
\end{eqnarray}
Therefore the $\rho _2(n)$ considered in this subsection grows more slowly with $n$ than the
$\rho _1(n)$ considered in the previous subsection.

\subsection{$\rho_3(n)=\frac {\Gamma (\alpha n+\beta)}{\Gamma (\beta)}$}\label{CC}
We consider the case
\begin{eqnarray}
\rho _3(n)=\frac {\Gamma (\alpha n+\beta)}{\Gamma (\beta)};\;\;\;\;\;\alpha, \beta > 0.
\end{eqnarray}
It was proved in \cite{F4} that
\begin{eqnarray}
{\cal N}_{\rho _3}(|z|^2 )=\Gamma (\beta)E_{\alpha, \beta}(|z|^2);\;\;\;\;\;\;
W_{\rho _3}(|z|^2)=\frac{|z|^{\frac{2(\beta-\alpha)}{\alpha}}}{\alpha \Gamma(\beta)}
\exp \left (-|z|^{\frac{2}{\alpha}}\right )
\end{eqnarray}
where $E_{\alpha, \beta}(y)$ are generalized Mittag-Leffler functions \cite{F4,G2}.
Therefore ${\mathfrak a}(\rho _3)=\alpha ^{-1}$ and ${\mathfrak b}(\rho _3)=1$.
In this case 
${\mathfrak B}_1(\alpha ^{-1},1) \subset {\cal B} (W_\rho) \subset {\mathfrak B}(\alpha ^{-1},1)$.

The set of coherent states $\ket {z_N;\rho _3}$
is overcomplete or undercomplete in the cases that
the density $(t, \delta)$ of the sequence $\{z_N\}$ is
$(t, \delta)\succ (\alpha ^{-1},\alpha ^{-1})$ or
$t <\alpha ^{-1}$, correspondingly.
For example , the sequence
\begin{eqnarray}
\zeta _N=\exp \left [s\ln (N)+i\theta _N\right]
\end{eqnarray}
where $\theta _N$ are arbitrary real numbers has density $(s^{-1},1)$ and
the set of  corresponding coherent states
is overcomplete [{\it resp.} undercomplete] when $s<\alpha$  or $s=\alpha >1$ [{\it resp.} $s>\alpha$].
Also, the sequence
\begin{eqnarray}\label{20A}
\zeta _N=\exp \left [\alpha \ln \left(\frac{N}{\delta}\right )+i\theta _N\right]
\end{eqnarray}
where $\theta _N$ are arbitrary real numbers has density $(\alpha ^{-1},\delta)$.
Therefore the set of  corresponding coherent states
is overcomplete when  $\delta >\alpha ^{-1}$.

\section{Discussion}

We have used the generalized coherent states studied in \cite{F1,F2,F3,F4,F5,F6,F7}
to define generalized Bargmann analytic representations in the complex plane.
For these states the scalar product in Eq.(\ref{12})  converges and this determines
the growth of the generalized Bargmann functions.
From the growth we infer the maximum density of the zeros of these functions and this in turn determines the overcompleteness or undercompleteness of discrete sets of these
generalized coherent states. In addition to the standard coherent states,
we studied three examples  in detail
in subsections \ref{AA}, \ref{BB} and \ref{CC}. 
Other examples in the same spirit can also be found.

The work provides a deeper insight into the use of the theory of analytic
functions in a quantum mechanical context.

{\bf Acknowledgement}: We thank the Agence Nationale de la Recherche (ANR) for support under the programme
PHYSCOMB No ANR-08-BLAN-0243-2.

\end{document}